\documentclass[letterpaper, 10 pt, conference]{ieeeconf}  % Comment this line out
\IEEEoverridecommandlockouts                              % This command is only
\usepackage{dsfont}
\usepackage{amsfonts}
\usepackage{amsmath}
\usepackage{amssymb}
\usepackage{mathtools}
\usepackage{pifont}
\DeclareMathOperator*{\argmax}{arg\,max}

\usepackage{cite}
\usepackage{booktabs}
\usepackage{multicol}
\usepackage{multirow}
\usepackage{tabularx}
\usepackage{diagbox}
\usepackage{graphicx}
\usepackage{subcaption}
\usepackage{array}
\usepackage{fancyhdr}
\usepackage{hyperref}
 \hypersetup{
     colorlinks=true,
     linkcolor=blue,
     filecolor=blue,
     citecolor = black,      
     urlcolor=cyan,
     }
\usepackage{cleveref}
\usepackage{xcolor}
\usepackage{tikz}
\usepackage{verbatim}
\usepackage{listings}
\definecolor{mygreen}{RGB}{28,172,0} % color values Red, Green, Blue
\definecolor{mylilas}{RGB}{170,55,241}

\usepackage{matlab-prettifier}
\usepackage{color}
\usepackage{algorithm}
\usepackage[noend]{algpseudocode}
\makeatletter
\def\BState{\State\hskip-\ALG@thistlm}
\makeatother
\newcommand{\tr}{\operatorname{Tr}}
\newcommand{\diag}{\operatorname{diag}}
\newcommand{\1}{\mathds{1}}
\newcommand{\Var}{\operatorname{Var}}
\newcommand{\E}{\mathbb{E}}

\newcommand{\tp}{\mathsf{T}}

\newcommand{\N}{\mathbb{N}}
\newcommand{\R}{\mathbb{R}}

\newcommand{\IF}{\text{if }}

\newcommand{\st}{\text{subject to}}
\newtheorem{theorem}{Theorem}
\newtheorem{lemma}{Lemma}
\newtheorem{corollary}{Corollary}
\newtheorem{proposition}{Proposition}
\newtheorem{example}{Example}
\newtheorem{definition}{Definition}
\newtheorem{remark}{Remark}
\newtheorem{assumption}{Assumption}

\DeclarePairedDelimiter\floor{\lfloor}{\rfloor}

\title{\LARGE \bf
On the Price of Transparency: A Comparison between Overt Persuasion and Covert Signaling
}

\author{Tao Li and Quanyan Zhu% <-this % stops a space
\thanks{Authors are with the Department of Electrical and Computer Engineering, New York University, NY, 11201, USA {\tt tl2636, qz494@nyu.edu}}% <-this % stops a space
}

\begin{document}

\maketitle
\thispagestyle{empty}
\pagestyle{empty}

\begin{abstract}
Transparency of information disclosure has always been considered an instrumental component of effective governance, accountability, and ethical behavior in any organization or system. However, a natural question follows: \emph{what is the cost or benefit of being transparent}, as one may suspect that transparency imposes additional constraints on the information structure, decreasing the maneuverability of the information provider. This work proposes and quantitatively investigates the \emph{price of transparency} (PoT) in strategic information disclosure by comparing the perfect Bayesian equilibrium payoffs under two representative information structures: overt persuasion and covert signaling models. PoT is defined as the ratio between the payoff outcomes in covert and overt interactions. As the main contribution, this work develops a two-stage-bilinear (TSB) programming approach to solve for non-degenerate perfect Bayesian equilibria of dynamic incomplete information games with finite states and actions. Using TSB, we show that it is always in the information provider's interest to choose the transparent information structure, as $0\leq \textrm{PoT}\leq 1$. The upper bound is attainable for any strictly Bayesian-posterior competitive games, of which zero-sum games are a particular case. For continuous games, the PoT, still upper-bounded by $1$, can be arbitrarily close to $0$, indicating the tightness of the lower bound. This tight lower bound suggests that the lack of transparency can result in significant loss for the provider. We corroborate our findings using quadratic games and numerical examples.        
\end{abstract}

\section{Introduction}
Information asymmetry is a prevailing phenomenon arising in a variety of contexts, such as in financial markets where insiders have more information than outsiders \cite{hayne77info-asy-finn}, in healthcare where doctors have more information about medical conditions than patients \cite{fabes22info-asy-health}, or in cyber deception where the defender has a better grasp of the enterprise network than the attacker \cite{jeff19deception-survey}. 

The double-edged nature of the imbalance of power caused by asymmetric information is noteworthy. On the one hand, it can foster the development of deception-based defense mechanisms that benefit the cybersecurity realm \cite{jeff19deception-survey}. On the other hand, it can also result in inefficiencies and exploitation in financial operations \cite{hayne77info-asy-finn}. Whether positive or negative, the imbalance of power is considered the root cause of unethical practices, even if the intention is benign (see a recent discussion on the ethics of cyber deception \cite{zhu23dce}).  

Addressing this ethical concern can be achieved by reducing information asymmetry through increased transparency. This can be accomplished by, for example, mandating companies to disclose more information to the public or by healthcare providers communicating medical information more clearly to patients. However, a natural question arises: \emph{what is the price of being transparent} the information provider (the sender) has to pay, as transparency requirements may impose additional constraints on the sender's side? Does increased transparency lead to decreased maneuverability for the sender, thereby impairing the effectiveness of systems built on information asymmetry, such as cyber deception in security applications? Does one have to choose between ethical standards and operational effectiveness?   

\begin{figure}
    \centering
    \includegraphics[width=0.4\textwidth]{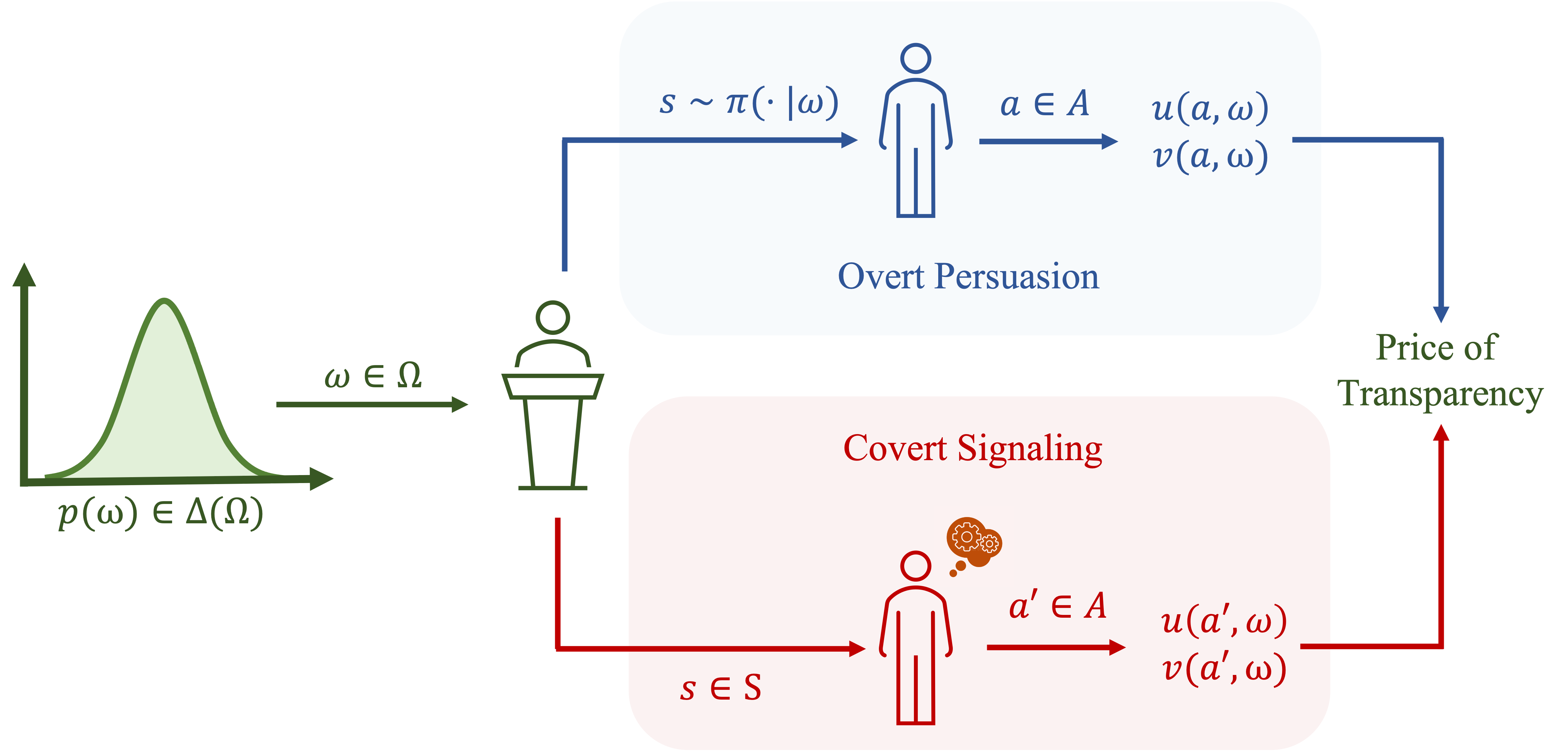}
    \caption{A schematic illustration of two juxtaposed communication games with distinct information structures. The state variable $\omega$ is randomly generated from the prior distribution $p$ and privately revealed to the sender. The receiver must infer the current state using the signal $s$ from the sender and then best respond to its belief. The payoff to the sender (the receiver), denoted by $u(a,\omega) [v(a,\omega)]$, is jointly determined by the state $\omega$ and the receiver's action $a$. The key difference between the two information structures is whether the signaling mechanism $\pi$ is made public or not. }
    \label{fig:two-info}\vspace{-6mm}
\end{figure}

As information asymmetry is prevalent in security applications \cite{tao23ztd} and other real-world systems \cite{tao22csm}, investigating the price of transparency (PoT) is imperative. This work initiates a quantitative study on PoT in strategic information disclosure. We consider a communication game between the sender and the receiver, where the state of nature is privately revealed to the sender only. Possessing this informational advantage, the sender discloses partial information (signal) regarding the state to the receiver to manipulate its belief, leading the receiver to take actions favored by the sender. The information disclosure mechanism (i.e., how the sender creates the signal) is referred to as the information structure in the literature \cite{tao_info}. To answer the questions above on the transparency of information disclosure, we compare two information structures: 1) overt persuasion (OP), where the sender publicly announces its mechanism, creating a transparent information disclosure, and 2) covert signaling (CS), where the mechanism is kept private throughout the gameplay, and the receiver only observes the signal. The two information structures are summarized in \Cref{fig:two-info}.

PoT is defined as the ratio of the sender's equilibrium payoff under covert signaling over its counterpart under overt persuasion to quantify the price of choosing the transparent information structure. Note that the equilibrium concept considered here is the perfect Bayesian equilibrium (PBE), as strategic information disclosure studied in this paper is a dynamic game of incomplete information, and players are assumed sequentially rational \cite{fudenberg}.      

\noindent\textbf{Main Results:}  As the transparency requirement mandates the sender to reveal its intention on signaling, it seems to give the receiver an upper hand. However, as opposed to the first impression, the key finding is that \text{PoT} $\leq1$ for any communication games, indicating that opting for the transparent information structure (OP) does not degrade the sender's payoff. On the contrary, the opaque one (CS) creates ``friction'' during the information transmission: the receiver needs to conjecture the sender's mechanism first, and then the conjecture must satisfy the consistency requirement in PBE. Consequently, covert signaling imposes more constraints on players' admissible strategies than overt persuasion. In comparison, transparent information disclosure leads to efficient communication, as players need not consider consistency.  

\noindent\textbf{Contributions:} Our main contributions include 1) the development of a two-stage-bilinear (TSB) programming approach (\Cref{thm:bb}) for solving non-degenerate PBE in strategic information disclosure; 2) the identification of a special class of communication games, termed strictly Bayesian-posterior games, for which the upper bound is attained: $\text{PoT}=1$, (\Cref{thm:upp}) ; 3) the construction of a family of quadratic games for which PoT can be arbitrarily close to $0$, indicating the tightness of the lower bound (\Cref{thm:lower}). 

To the best of our knowledge, this work is among the first endeavors to characterize PBE using bilinear programming. Solving for PBE is challenging, as the belief consistency in PBE involves Hadamard division. The resulting problem is highly nonlinear. By transferring the problem into the posterior-belief space, the proposed $Z$-programming creates a viable approach for computing PBE, a prevalent solution concept in dynamic games of incomplete information.  This work can shed light on many related social and engineering problems with asymmetric information.       

\noindent\textbf{Related Works:} This work stands at the intersection of two lines of research: strategic information transmission and algorithmic information design. The two information structures are inherited from the Bayesian persuasion model in \cite{kamenica11BP} and the signaling model in \cite{crawford82signaling}, respectively. Starting from these seminal models in strategic information transmission, this work carries out a comparative study of the two information structures. Unlike early comparative studies \cite{gossner00compare,green07compare} focusing on Bayesian Nash equilibrium, this work treats perfect Bayesian equilibrium, a more challenging concept involving belief consistency.  

This work also subscribes to the recent line of works that explores the computational aspect of information structure design \cite{rubinstein15hardness,bhaskar16hardness,dughmi19hardness,tao22bp}. These mentioned works provide hardness results on the computational complexity of solving for the equilibrium information structures without showing concrete algorithms. In contrast, not concerning the existence of PBE or the associated complexity, we present a bilinear programming approach to compute PBE, which in turn corroborates these hardness results in \cite{rubinstein15hardness,bhaskar16hardness,dughmi19hardness,tao22bp}.

\section{Strategic Information Disclosure: Persuasion and Signaling}\label{sec:model}
Consider a communication game as in \cite{crawford82signaling}, where the better-informed sender, upon receiving the state of nature, sends a signal to the receiver who then takes an action that determines payoffs to both players. Mathematically, the game model is given by a tuple $\left\langle \Omega, \mathcal{S}, \mathcal{A},v, u\right\rangle$, where 1) $\Omega$ is the set of possible states with its typical element denoted by $\omega$, and the realization of $\omega$ is only revealed to the sender; 2) $p\in\Delta(\Omega)$ denotes the prior distribution over the state space; 3) $\mathcal{S}$ is the set of signals possessed by the sender with its typical element denoted by $s$; 4) $\mathcal{A}$ is the action space of the receiver; 5) $u, v: \Omega\times\mathcal{A}\rightarrow \R $ are utilities\footnote{Unless otherwise specified, $u,v\geq 0$.} of the sender and the receiver, respectively.  $\Delta(\cdot)$ denotes the set of all probability measures compatible with the underlying $\sigma$-algebra (e.g., Borel) over the set of interest.

\noindent\textbf{Information Structures:} The information structure concerns how the sender signals to  the receiver. An information structure or signaling mechanism is defined by a mapping $\pi: \Omega\rightarrow\Delta(\mathcal{S})$, i.e., $\pi(\cdot|\omega)$ is a probability distribution over the signal space. As detailed below, the difference between overt persuasion and covert signaling is whether $\pi$ is revealed to the receiver before it acts.

\noindent\textbf{Covert Signaling (CS):} As shown in \Cref{fig:two-info}, the information structure $\pi$ is unknown to the receiver who consequently cannot  form a posterior belief $\lambda(\cdot|s)\in \Delta(\Omega)$, as the Bayes update requires the knowledge of the information structure $\pi$: $\lambda(\omega|s)=\frac{\pi(s|\omega)p(\omega)}{\int_{\Omega}\pi(s|\omega')p(d\omega')}$. In this case, the receiver can begin with a conjectural belief system, and best respond to these beliefs. If the belief system is consistent with the receiver's and the sender's strategies in the Bayesian sense, then the belief system and players' strategies constitute a perfect Bayesian equilibrium.
\begin{definition}[Perfect Bayesian Equilibrium \cite{crawford82signaling}]\label{def:pbe}
For a communication game, a triple of the sender's information structure $\pi$, the receiver's strategy $\alpha: \mathcal{S}\rightarrow \Delta(\mathcal{A})$, and a belief system $\lambda: \mathcal{S}\rightarrow\Delta(\Omega)$ is a perfect Bayesian equilibrium if it satisfies 1) among all admissible information structures, $\pi$ maximizes the sender's expected payoff given $\alpha$, see \eqref{eq:sender-br}; 2) for any signal $s$, $\alpha(s)$ maximizes the receiver's expected payoff under the belief system $\lambda$, see \eqref{eq:receiver-br}; 3) the belief system  is consistent with $\pi$ and $\alpha$, see \eqref{eq:consistency}.
\begin{align}
&\pi \in \argmax_{\mu: \Omega\rightarrow \Delta(\mathcal{S})} \int_{\Omega} \int_{\mathcal{S}}\int_{\mathcal{A}} u(a,\omega)\alpha(da|s) \mu(ds|\omega) p(d\omega), \label{eq:sender-br}\\
& \alpha(\cdot|s) \in \argmax_{\mu \in \Delta(\mathcal{A})} \int_{\mathcal{A}} v(a,\omega)\mu(da) \lambda(d\omega|s), \label{eq:receiver-br}\\
& \lambda(\omega|s)=\left\{\begin{array}{ll}
     \frac{\pi(s|\omega)p(\omega)}{\int_{\Omega}\pi(s|\omega')p(d\omega')}, \text{ if } \int_{\Omega}\pi(s|\omega')p(d\omega')>0  \\
      \text{any probability distribution on }\Omega, \text{ otherwise}.
\end{array}\right. \label{eq:consistency}
\end{align}
\end{definition}
Let $\gamma(s)=\int_{\Omega}\pi(s|\omega')p(d\omega')$ denote the probability of generating a particular signal $s$, and we refer to $\operatorname{supp}(\gamma)$ as the set of realizable signals. The consistency  \eqref{eq:consistency} in PBE requires that for any realizable signal, the belief system conforms to the Bayes rule under equilibrium strategies $\pi$ and $\alpha$. For those unrealizable (the receiver never observes these signals), the beliefs can be arbitrary ones, as they never appear on the equilibrium path \cite[Chapter 8.2]{fudenberg}, which makes no difference to the equilibrium strategies. Yet, these arbitrary beliefs may cause trouble in our bilinear programming formulation in finite games presented in \Cref{sec:pot}. Hence, we restrict the focus to a subset of PBE, where every signal is realizable.
 \begin{definition}[Non-degenerate PBE]
\label{def:non-pbe}
For a triple of $\pi$, $\alpha$, and $\lambda$, if it satisfies \eqref{eq:sender-br}, \eqref{eq:receiver-br}, \eqref{eq:consistency}, and $\operatorname{supp}(\gamma)=\mathcal{S}$, then it is a non-degenerate PBE: every signal gives rise to an equilibrium path. 
 \end{definition}

\noindent\textbf{Overt Persuasion (OP):} Unlike covert signaling, the sender in overt persuasion first reveals $\pi$ to the receiver and then draws a signal according to $\pi(\cdot|\omega)$, when the realized state is $\omega$. Hence, the receiver need not conjecture, as the belief $\lambda(\omega|s)$ is readily available through Bayesian update once the signal is observed. The sender's equilibrium payoff is determined by backward induction. The equilibrium information structure is given by\footnote{To simplify the exposition, we only consider the non-degenerate case. For degenerate cases, take $\lambda$ as an arbitrary distribution.}
\begin{equation}
    \begin{aligned}
    \max_{\pi:\Omega\rightarrow \Delta(\mathcal{S})} & \int_{\Omega} \int_{\mathcal{S}}\int_{\mathcal{A}} u(a,\omega)\alpha(da|s) \pi(ds|\omega) p(d\omega)\\
    \text{subejct to } & \alpha(\cdot|s) \in \argmax_{\mu \in \Delta(\mathcal{A})} \int_{\mathcal{A}} v(a,\omega)\mu(da) \lambda(d\omega|s),\\
    & \lambda(\omega|s)=\frac{\pi(s|\omega)p(\omega)}{\int_{\Omega}\pi(s|\omega')p(d\omega')}.
\end{aligned}\label{eq:op}
\end{equation}

\begin{remark}
Note that due to the transparent information structure in OP, there is no consistency requirement as in \eqref{eq:consistency}, and only two perfection conditions are left: \eqref{eq:sender-br} and \eqref{eq:receiver-br}. Hence, the perfect equilibrium becomes a sender-preferred subgame perfect equilibrium (SPE) \cite{kamenica11BP}, which can be solved using backward induction. In contrast, PBE carries a circularity: the beliefs are consistent with the strategies, which are optimal given the beliefs. Even if players move sequentially in CS, PBE cannot be determined by backward induction. 
\end{remark}

\noindent\textbf{Bayesian Plausibility and Backward Induction:}
Since \eqref{eq:op} is a bilevel optimization of functionals, directly solving for $\pi$ remains challenging. A key observation in \cite{kamenica11BP} is that an information structure is equivalent to a distribution over posterior beliefs. Recall that  each signal $s$ in OP leads to a posterior belief $\lambda(s)\in \Delta(\Omega)$ with respect to the information structure $\pi$. Accordingly, each information structure $\pi$ leads to a distribution over posterior beliefs. Denote a distribution of posteriors by $\tau\in \Delta(\Delta(\Omega))$, and $\tau$ is given by $\tau(\lambda )=\int_{s:\lambda=\lambda(\cdot|s)}\int_{\Omega}\pi(s|\omega)p(d\omega)ds$, assuming that $\{s:\lambda=\lambda(\cdot|s)\}$ is a measurable subset of $\mathcal{S}$, and $ds$ denotes the Borel measure. This work considers the cases where $\Omega$ is a separable metric space (e.g., $\R^n$ or finite sets). Consequently, $\Delta(\Omega)$, endowed with weak$^*$-topology, is also separable and metrizable. Hence, the Borel probability measure is well-defined on $\Delta(\Omega)$.    

With a slight abuse of notation, we also denote by $\lambda$ an individual belief in $\Delta(\Omega)$. A belief is Bayesian inducible under $\pi$ if $\tau(\lambda )>0$, i.e., $\lambda\in \operatorname{supp}(\tau)$, and distribution of posteriors $\tau$ is \emph{Bayesian plausible} if the expected posterior probability equals the prior: $\int_{\operatorname{supp}(\tau)}\lambda \tau(d\lambda )=p$. \cite{kamenica11BP} finds that for any Bayesian plausible distribution $\tau$, one can always find an information structure $\pi$ such that every $\lambda\in \operatorname{supp}(\tau)$ is $\pi$-Bayesian inducible. With this observation, searching for the optimal information structure is equivalent to finding the optimal posteriors distribution through backward induction specified below.

Given a posterior belief $\lambda$, denote by $\hat{a}(\lambda )=\argmax_{a}\E_{\omega\sim \lambda}v(a,\omega)$ the best response of the receiver, which is assumed to be a singleton (tie breaks in favor of the sender)\footnote{If the belief $\lambda$ is induced by the signal $s$, then $\alpha(s)=\hat{a}(\lambda)$}. Under this belief, the sender's expected utility is $\hat{u}(\lambda )= \E_\lambda u(\hat{a}(\lambda ),\omega)$. If $\lambda $ is further subject to a distribution $\tau$, then the sender's payoff is $\E_\tau \hat{u}(\lambda )$. Since the sender's goal is to find the distribution $\tau$ that maximizes his expected utility, the corresponding optimization problem is given by 
\begin{equation}
	\begin{aligned}
	\max_\tau \quad & \E_\tau \hat{u}(\lambda )\\
	\st & \int_{\lambda\in \operatorname{supp}(\tau)}\lambda \tau(d\lambda )=p.
\end{aligned} \label{eq:spe}
\end{equation}
\begin{remark}
Bayesian plausibility transfers the problem to the space of posterior beliefs, paving the way for further analysis of \eqref{eq:spe}. For example, one can prove that $\hat{u}(\lambda)$ is upper semicontinuous \cite{kamenica11BP}. Several hardness results in computing or learning the information structure \cite{dughmi19hardness,tao22bp} rest on this plausibility. As later shown in \Cref{sec:z}, Bayesian plausibility also plays a significant role in our programming.  
\end{remark}
\textbf{The Price of Transparency} 
\begin{definition}[Price of Transparency]
    Denote by $U^{CS}$ and $U^{OP}$ the sender's equilibrium payoff in covert signaling and overt persuasion, respectively. The price of transparency (PoT) is defined as $\text{PoT}=\frac{U^{CS}}{U^{OP}}$.
\end{definition}
 Note that $U^{OP}$, the optimal value in \eqref{eq:op} [or equivalently \eqref{eq:spe}], is unique. In contrast, the communication game in CS may admit multiple equilibria and hence, different equilibrium payoffs. To understand this subtle difference between OP and CS, one can think of CS as a Nash play, as the best response conditions \eqref{eq:sender-br} and \eqref{eq:receiver-br} hold simultaneously. Thanks to the sender's commitment to a transparent information structure, the SPE in \eqref{eq:op} is in the same vein as Stackelberg games. These observations are more evident from finite games discussed in \Cref{sec:pot}. Given that PoT is not a definite number but a collection of possibilities, we aim to identify its upper and lower bounds for the rest of the paper.        
\section{The Price of Transparency in Finite Games}\label{sec:pot}
\textbf{Matrix Representation of Information Structure} Our treatment of PoT begins with finite games where $\Omega, \mathcal{S}$ and $\mathcal{A}$ are all finite discrete sets. In finite games, the sender's and the receiver's strategies and the belief system all take matrix forms. We introduce some notations in the following to facilitate the discussion. Let $\Omega=\{\omega_i\}_{i\in [M]}, \mathcal{S}=\{s_i\}_{i\in [N]}$, and $\mathcal{A}=\{a_i\}_{i\in [K]}$, where $[n]:=\{1,2,\ldots,n\}, n\in \N_{+}$. Assume that $N\geq M$. Denote by $p\in \R^{M}$ the prior distribution over $\Omega$, and by $U=[U_{km}=u(a_k,\omega_m)]\in \R^{K\times M}, V=[V_{km}=v(a_k,\omega_m)]\in \R^{K\times M}$ the sender's and the receiver's utilities, respectively. The sender's information structure is specified by a right stochastic matrix $\Pi=[\Pi_{mn}=\pi(s_n|\omega_m)]\in \R^{M\times N}$. The receiver's strategy is given by a right stochastic matrix $A=[A_{nk}=\alpha(a_k|s_n)]\in \R^{N\times K}$. Denote by $\mathds{1}$ the all-one vector of a proper dimension depending on the context, and then $\Pi\1=\1, A\1=\1$. 

\noindent\textbf{Notations:} In addition to the above, other helpful notations are as follows. $e_i$ refers to the $i$-th elementary vector of a proper dimension depending on the context. $J_{K\times M}$ denotes the $K\times M$ matrix of 1's. For a vector $w$, $\diag(w)$ denotes the diagonal matrix with $w$ on its diagonal. For a square matrix $W$, $\diag(W)$ denotes the vector containing its diagonal entries.  For any two vectors $w,v\in \R^N$ of the same dimension, $w\succeq v$ (or $w\succ v$) indicates entry-wise relations: $w_i\geq v_i, \forall i\in [N]$. $\oslash$ denotes the Hadamard division (entry-wise): $w\oslash v=[w_i/v_i]_{i\in [N]}$. $\tr(W)$ denotes the trace of a square matrix $W$.  $W_j$ refers to the $j$-th column, and its transpose of $W$ is denoted by $W^\tp$, while  $W'$ denotes its perturbation within the same domain specified by the context. 

Define the prior matrix as $P=\diag(p)\in \R^{M\times M}$. Given the players' strategies $\Pi$ and $A$, the sender's expected payoff is $\sum_{m}p_m\sum_{n}\Pi_{mn}\sum_{k}A_{nk}U_{km}=\tr(P\Pi A U)$. Under the information structure $\Pi$, the receiver's posterior belief upon observing signal $s_n$ is $\lambda_{mn}=\frac{p_m \Pi_{mn}}{\sum_{m'} p_{m'}\Pi_{m'n}}$. Define the belief system as $\Lambda=[\lambda_{mn}]\in \R^{M\times N}$. According to the Bayes rule shown above, the information structure and the belief system satisfies $\Lambda=P\Pi \oslash ( \mathds{1}\mathds{1}^\tp P\Pi)$. The receiver's strategy $A$ is a best response to the posterior belief $\Lambda$, i.e., $\sum_{m}\lambda_{mn}\sum_{k}A_{nk}V_{km}\geq \sum_{m}\lambda_{mn}\sum_{k}A'_{nk}V_{km}$, for any $n\in [N]$, and any right stochastic matrix $A'$.  Summing up all the equations above, we arrive at the following statement.
\begin{proposition}[PBE in Matrix Form]\label{prop:pbe-matrix}
For a finite communication game, a triple of matrices $(\Pi, A, \Lambda)$ is a perfect Bayesian equilibrium if it satisfies  
\begin{align}
    &\operatorname{Tr}(P\Pi A U)\geq \operatorname{Tr}(P \Pi' A U), \forall \Pi'\in \R_{\geq 0}^{M\times N}, \Pi'\mathds{1}=\mathds{1},\label{eq:se-br}\\
    & \operatorname{diag}(AV\Lambda)\succeq \operatorname{diag}(A'V\Lambda ),\forall A'\in \R_{\geq 0}^{N\times K}, A'\mathds{1}=\mathds{1},\label{eq:re-br}\\
    & \Lambda=P\Pi \oslash ( \mathds{1}\mathds{1}^\tp P\Pi),\label{eq:lambda}\\
    & \Pi\mathds{1}=\mathds{1}, A\mathds{1}=\mathds{1}, \Pi\in \R_{\geq 0}^{M\times N}, A\in \R_{\geq 0}^{N\times K}. \nonumber
\end{align}
\end{proposition}
Several remarks are in order. First, it is straightforward to see the one-to-one correspondence between equations in \Cref{prop:pbe-matrix} and those in \Cref{def:pbe}. \eqref{eq:se-br} corresponds to \eqref{eq:sender-br}, ensuring $\Pi$ is the best response to $A$. Similar to \eqref{eq:receiver-br}, \eqref{eq:re-br} asserts the optimality of $A$ given $\Lambda$ that is consistent with $\Pi$, as enforced by \eqref{eq:lambda}. Second, \Cref{prop:pbe-matrix} clearly demonstrates that solving for PBE is challenging, as Hadamard division in the belief system makes the problem highly nonlinear. Fortunately, this nonlinearity created by Hadamard division can be bypassed using Bayesian plausibility for non-degenerate PBE, as shown later in \Cref{sec:z}. 

Finally, we conclude this section with the matrix representation of SPE in \eqref{eq:op}, from which we draw an analogy between SPE and Stackelberg equilibrium. 
\begin{proposition}[SPE in Matrix Form]\label{prop:spe-matrix}
    For a finite communication game, a pair of matrices $(\Pi, A)$ is a sender-preferred subgame perfect equilibrium  if it satisfies 
    \begin{equation}
    \begin{aligned}
        \max_{\Pi, A}\quad & \operatorname{Tr}(P\Pi A U) \\
 	\st \quad & \operatorname{diag}(AV\Lambda)\succeq \operatorname{diag}(A'V\Lambda ),\\
  &\forall A'\in \R_{\geq 0}^{N\times K}, A'\mathds{1}=\mathds{1},\\
 	& \Lambda^\tp=P\Pi \oslash ( \mathds{1}\mathds{1}^\tp P\Pi),\\
 	& \Pi\mathds{1}=\mathds{1}, A\mathds{1}=\mathds{1}, \Pi\in \R_{\geq 0}^{M\times N}, A\in \R_{\geq 0}^{N\times K}.
    \end{aligned}\label{eq:spe-matrix}
    \end{equation}
\end{proposition}
Despite the belief system, \eqref{eq:spe-matrix} takes the same bilevel optimization formulation as Stackelberg equilibrium, as both equilibrium concepts require sequential rationality \cite{fudenberg}.  In contrast, \eqref{eq:se-br} and \eqref{eq:re-br} resemble the best response condition in Nash equilibrium. Given this observation, the statement that PoT is less than $1$ becomes foreseeable since comparative studies on Nash and Stackelberg equilibrium have already arrived at similar conclusions \cite{simaan77nash-stackelberg}.      
\section{Bayesian Plausibility and TSB-Programming}\label{sec:z}
The above discussion, though intuitive, does not reveal the tightness of the upper bound. This section develops the TSB programming approach to solve for PBE in \Cref{prop:pbe-matrix}, further enabling us to prove that PoT is tightly upper bounded by 1. As the existence of PBE remains an open question\cite{mensch20pbe}, we impose a standing assumption on the existence of non-degenerate PBE to secure the well-posedness of the proposed programming. 
\begin{assumption}\label{ass:exist}
For any finite communication games in this work, there exists at least one non-degenerate PBE.  
\end{assumption}

Recall that Bayesian plausibility is a sanity check for an information structure: all possible posterior beliefs should be consistent with the prior under the information structure. Using mathematical terms, $ p=\sum_{n}\gamma_n\Lambda_n$, and $\gamma_n=\sum_{m}p_m\Pi_{mn}$ denotes the probability of generating $s_n$, which is a discrete counterpart to $\gamma(s)$ defined in \Cref{sec:model}. For non-degenerate PBE, $\gamma_n>0$ for all $n\in [N]$. 

Note that $\lambda_{mn}=\frac{p_m \Pi_{mn}}{\sum_{m'} p_{m'}\Pi_{m'n}}$, then $p_m\Pi_{mn}=\lambda_{mn}\gamma_n$. Hence, the sender's expected payoff can be rewritten using posterior beliefs $\Lambda$ and $\Gamma:=\diag(\gamma)$, as shown below:
\begin{align*}
    \tr(P\Pi A U)&=\sum_{m}p_m\sum_{n}\Pi_{mn}\sum_{k}A_{nk}U_{km}\\
    &=\sum_{n}\sum_{k}\sum_{m}A_{nk}U_{km}p_m\Pi_{mn}\\
    &=\sum_{n}\sum_{k}\sum_{m}A_{nk}U_{km} \lambda_{mn}\gamma_n = \tr(AU\Lambda\Gamma).
\end{align*}
The above deduction actually gives an elementary proof of the one-to-one correspondence between information structure and the posterior distribution we discussed in \eqref{eq:spe}. Meanwhile, as $\gamma\succ 0$, then $ \operatorname{diag}(AV\Lambda)\succeq \operatorname{diag}(A'V\Lambda ) \Leftrightarrow \operatorname{diag}(AV\Lambda\Gamma)\succeq \operatorname{diag}(A'V\Lambda \Gamma )$. Finally, one can see that both the sender's and the receiver's best response conditions involve the matrix product of  $\Lambda$ and  $\Gamma$, creating another matrix representation of PBE presented in \Cref{thm:z}.  

Define $Z=\Lambda\Gamma\in \R^{M\times N}$, then according to Bayesian plausibility, $Z\1=\Lambda \Gamma\1=\Lambda \gamma=p$. Another constraint on $Z$ arises from the left stochasticity of $\Lambda$, i.e., $\1^\tp \Lambda = \1^\tp $, implying that $\1^\tp Z=1^\tp \Lambda \Gamma = 1^\tp \Gamma=\gamma^\tp \succ 0$. Hence, $0 \prec Z^\tp\1 \prec \1$. Summarizing these constraints, we denote by $\mathcal{Z}:= \{Z| Z\in \R_{\geq 0}^{M\times N}, Z\1=p, 0 \prec Z^\tp\1 \prec \1 \}$ the set of Bayesian plausible matrices.
We then arrive at the following theorem where PBE is characterized using the $Z$ matrix. 
\begin{theorem}\label{thm:z}
    For a finite communication game, a pair matrices of $(Z,A)$ is a non-degenerate perfect Bayesian equilibrium if it satisfies 
    \begin{equation}
        \begin{aligned}
    & \operatorname{Tr}(AUZ)\geq \operatorname{Tr}(AUZ'), \forall Z'\in \mathcal{Z},\\
     & \operatorname{diag}(AVZ) \succeq  \operatorname{diag}(A'VZ), \forall A'\in \R_{\geq 0}^{N\times K}, A'\mathds{1}=\mathds{1},\\
     & A^\tp\1=\1, Z\in \mathcal{Z}.
        \end{aligned} \label{eq:z}
    \end{equation}
\end{theorem}
\begin{proof}
The proof proceeds in the reverse direction of the deductions above. If $Z$ and $A$ are the solution pair, then let $\gamma=Z^\tp \1$, and $\Lambda_n=Z_n\oslash \gamma_n$. It is easy to see that $\1^\tp \gamma =\1^\tp Z^\tp \1=p^\tp \1=\1$ and $\1^\tp \Lambda_n=\1^\tp Z_n\oslash \gamma_n=\1$, implying that $(\gamma, \Lambda)$ constitutes a valid pair of the posterior distribution and the associated belief system. Hence, the corresponding $\Pi$ satisfies \eqref{eq:se-br}, and naturally $A$ satisfies \eqref{eq:re-br} ($\gamma\succ 0$), which concludes the proof.    
\end{proof}
The significance of \Cref{thm:z} is self-evident: no Hadamad division is involved, and \eqref{eq:z} is a constrained bilinear programming with respect to $Z$ and $A$. Intuitively, $Z$ matrix transfers the equilibrium problem into the posterior belief space, eliminating the nonlinearity introduced by Bayes rule [see \eqref{eq:lambda}]. The remaining nonlinearity (bilinearity) emerges from the nature of Nash play \cite{basar98game}.  

\noindent\textbf{Belief-Dominant Equilibrium:}
Even though \Cref{thm:z} seems to be the light at the end of the tunnel, directly solving PBE using \eqref{eq:z} is still daunting. It is cumbersome to vectorize $Z$ and $A$ and then transform \eqref{eq:z} into a standard bilinear form. Since the vectorization concatenates row or column vectors of $Z$ and $A$, the resulting vectors are no longer stochastic vectors, rendering many techniques in bilinear programming \cite[Chapter 3.4]{basar98game} inapplicable.  The following presents an equivalence between \eqref{eq:z} and a two-stage-bilinear programming, where the optimal solutions to the first stage problem constitute the feasible set of the second stage programming.      

Recall that the first inequality in \eqref{eq:z} gives 
\begin{align*}
    \sum_{i\in [N]}a_i^\tp U z_i = \operatorname{Tr}(AUZ) \geq \operatorname{Tr}(AUZ') = \sum_{i\in [N]}a_i^\tp U z'_i,
\end{align*}
where $a_i(a_i')$ and $z_i(z_i')$ are the $i$-th row vector of $A(A')$ and $i-$th column vector of $Z(Z')$, respectively.  The trace inequality in \eqref{eq:z} implies that the sender does not deviate from the equilibrium belief system $\Lambda$ and the distribution $\Gamma$, as the resulting average payoff over every signal is optimal. Note that $z_i=\gamma_i\lambda_i$, $\lambda_i$ is the $i$-th column of $\Lambda$, and consider the following constraints: $ a_i^\tp U \lambda_i\geq a_i^\tp U \lambda'_i$, for any $i\in [N]$, $\lambda'_i\in \Delta([N])$. Compared to the trace, the newly introduced ones require the equilibrium belief itself to be optimal for each signal, fixing the receiver's move. The latter is stronger than the former. We refer to the PBE characterized by the stronger constraints as belief-dominant PBE, as the belief system $\Lambda$ best responds to $A$ and dominates all other belief systems.   
\begin{definition}[Belief-Dominant PBE]
    A non-degenerate PBE $(Z=\Lambda \Gamma,A)$ is said to be belief-dominant, if the belief system $\Lambda$ best responds to $A$ for each signal: $ a_i^\tp U \lambda_i\geq a_i^\tp U \lambda'_i$, $a_i=(A^\tp)_i$, $\forall i\in [N]$, $\forall \lambda'_i\in \Delta([N])$.
\end{definition}
The notion of belief dominance only applies to non-degenerate PBE, as the belief vector $\lambda_i$ can be arbitrary when $\gamma_i=0$ [see \eqref{eq:consistency}] in degenerate cases. The following presents an example of belief-dominant PBE.  
\begin{example}[Non-degenerate and Belief-Dominant PBE]
\label{exam:bd-pbe}
 Consider a two-state, two-action, and two-signal case: $M=N=K=2$. The sender's and the receiver's utility matrices are 
 \begin{align*}
    U= \begin{bmatrix}
  1& 0\\
  0 & 1/2
 \end{bmatrix}, \quad V= \begin{bmatrix} 1 & 0\\
 1 &2 
 \end{bmatrix}.
 \end{align*}
The prior is $p=(1/2, 1/2)$. Both parties prefer $a_2$ in state $\omega_2$. The sender prefers $a_1$ in state $\omega_1$, while the receiver is indifferent between two actions. 

As shown in the utility matrices, the interests of the sender and the receiver are aligned. Hence, one special perfect Bayesian equilibrium strategy for the sender is the so-called truth-telling strategy: $\Pi=I$, also referred to as the separating equilibrium \cite[Chapter 8]{fudenberg}. As a non-degenerate PBE, the separating equilibrium is given by 
\begin{align*}
 & \Pi=\Lambda=\begin{bmatrix}
     1 & 0\\
     0 & 1
     \end{bmatrix},
    \Gamma  =\begin{bmatrix}
     1/2 & 0\\
     0 & 1/2
     \end{bmatrix},
     A = \begin{bmatrix}
         1-\epsilon & \epsilon\\
     0 & 1
     \end{bmatrix},
\end{align*}
where $\epsilon\in [0,1]$. The $\epsilon$ entries in $A$ are due to the fact that the receiver is indifferent between two actions when it observes $s_1$ and realizes that the state is $\omega_1$. In other words, there is a continuum of non-degenerate PBE. Direct calculation gives that those PBE with $\epsilon\in [0, 2/3]$ is belief-dominant.    
\end{example}
\begin{remark}
    The existence of belief-dominant PBE is demonstrated in the above example, even though a general existence guarantee remains unclear. This work does not attempt to answer this question, as it is beyond the scope. Instead, we assume the existence holds for the communication games. Another remark is that belief dominance yields a subset of non-degenerate PBE, for which \eqref{eq:z} admits a simpler formulation detailed below. Yet, \eqref{eq:z} itself returns every non-degenerate PBE. One can verify that the following non-dominant equilibria also solve the programming problem: 
    \begin{align*}
         Z=\begin{bmatrix}
         \frac{1}{2} & 0 \\
         \frac{\epsilon}{2} & \frac{1-\epsilon}{2}
     \end{bmatrix},
     \Lambda =\begin{bmatrix}
         \frac{1}{1+\epsilon} & 0\\
 \frac{\epsilon}{1+\epsilon} & 1
     \end{bmatrix},
     A=\begin{bmatrix}
 0 & 1\\
 0 & 1
 \end{bmatrix}, \epsilon\in [0,1).
    \end{align*}
\end{remark}

\textbf{Two-Stage-Bilinear Programming} For belief-dominant PBE, the constraints in \eqref{eq:z} reduces to 
\begin{equation}
    \begin{aligned}
        &a_i^\tp U \lambda_i\geq a_i^\tp U \lambda'_i, \lambda_i^\tp \mathds{1}=\lambda_i'^\tp \mathds{1}=1, \lambda_i, \lambda_i'\geq 0,\\
        &a_i^\tp  V \lambda_i \geq a_i'^\tp V \lambda_i, a_i^\tp \mathds{1}=a_i'^\tp \mathds{1}=1, a_i,a_i'\geq 0,
    \end{aligned}\label{eq:nash-constraint}
\end{equation}
 for any $i\in [N]$. The benefit of considering these constraints is that all decision variables involved are from the probability simplex, leading to the helpful lemma below. 
 \begin{lemma}
 \label{lem:bilinear}
 A pair  $(\{a_i\},\{\lambda_i\})$ constitutes a feasible point to \eqref{eq:nash-constraint} if and only if there exists a pair $(\{x_i\}, \{y_i\})$ such that $(\{a_i\}, \{\lambda_i\}, \{x_i\}, \{y_i\})$ solves the  bilinear programming
 \begin{equation}
    \begin{aligned}
    \max_{a_i,\lambda_i, x_i,y_i} & \sum_{i\in [N]} a_i^\tp U \lambda_i +\sum_{i\in [N]} a_i^\tp V \lambda_i -\sum_{i\in [N]}x_i- \sum_{i\in [N]} y_i \\
    \text{subject to }&  U^\tp a_i \preceq x_i \mathds{1}, V \lambda_i \preceq y_i \mathds{1},\\
    & a_i\geq 0, \lambda_i\geq 0, a_i^\tp \mathds{1}=1, \lambda_i^\tp \mathds{1}=1.
\end{aligned} \label{eq:bilinear}   
 \end{equation}
 \end{lemma}
\begin{proof}
    For any feasible $\{a_i\},\{\lambda_i\}$ to \eqref{eq:bilinear}, the constraints evidently imply that $a_i^\tp V \lambda_i\leq y_i$ and $a_i^\tp U \lambda_i \leq x_i$. Hence, $\sum_{i} a_i^\tp U \lambda_i +\sum_{i} a_i^\tp V \lambda_i -\sum_{i}x_i- \sum_{i} y_i\leq 0$, showing that the objective value is non-positive. If $\{a_i^*\}$ and $\{\lambda_i^*\}$ are the equilibrium pair to \eqref{eq:nash-constraint},then let $y_i^*=a_i^{*\tp} V \lambda_i^*$, $x_i^*=a_i^{*\tp} U \lambda_i^*$. Then, the resulting quadruple is feasible, and the corresponding value is zero, indicating optimality. 

    Conversely, let $(\{\bar{a}_i\}, \{\bar{\lambda}_i\}, \{\bar{x}_i\}, \{\bar{y}_i\})$ be a solution to \eqref{eq:bilinear}.  From the constraints, we have $V\bar{\lambda}_i\preceq\bar{y}_i \mathds{1}$ and $U^\tp \bar{a}_i \preceq \bar{x}_i \mathds{1}$.  Hence, for any $a_i, \lambda_i$, $a_i^\tp V \bar{\lambda}_i \leq \bar{y}_i$, $\bar{a}_i^\tp U \lambda_i \leq \bar{x}_i$. Likewise, $\bar{a}_i^\tp V \bar{\lambda}_i \leq \bar{y}_i$, $\bar{a}_i^\tp U \bar{\lambda}_i \leq \bar{x}_i$. Note that \Cref{ass:exist} guarantees the existence of PBE, at which the objective value in \eqref{eq:bilinear} is zero. Hence, zero is attainable: $\sum_{i\in [N]} \bar{a}_i^\tp U \bar{\lambda}_i +\sum_{i\in [N]} \bar{a}_i^\tp V \bar{z}_i - \sum_{i} \bar{x}_i -\sum_{i}\bar{y}_i=0$. Then, for the solution pairs, $\bar{a}_i^\tp U \bar{\lambda}_i=\bar{x}_i\geq a_i^\tp U \lambda_i$, $\bar{a}_i^\tp V \bar{\lambda}_i = \bar{y}_i\geq a_i^\tp V \bar{\lambda}_i$.
\end{proof}
\begin{corollary}
\label{coro:signal-payoff}
For any solution $(\{a_i\}, \{\lambda_i\}, \{x_i\}, \{y_i\})$  to the bilinear programming \eqref{eq:bilinear}, if there exists a $\gamma\in \Delta([M])$ such that $\Lambda\gamma=p$, $\Lambda_i=\lambda_i$, then $(\Lambda, \gamma)$ is Bayesian plausible. Under the corresponding information structure, the sender's (receiver's) expected payoff under signal $s_i$ is $x_i$ ($y_i$).   
\end{corollary}
The proof of \Cref{coro:signal-payoff} rests on the fact that the solution quadruple to \eqref{eq:bilinear} satisfies the equations: $a_i^\tp U \lambda_i=x_i$, $a_i^\tp V \lambda_i=y_i$. As one can see from \Cref{coro:signal-payoff}, an optimal posterior distribution $\gamma$ should maximize the expected payoff of all signals: $\sum_{i}\gamma_ix_i$. This observation leads to \Cref{thm:bb}, where solutions to the bilinear programming in \eqref{eq:bilinear} constitute a feasible set to another bilinear programming.
\begin{theorem}[Two-Stage-Bilinear Programming]
\label{thm:bb}
For any solution $(Z=(\Lambda,\gamma), A)$ to \eqref{eq:z} that is belief-dominant, it is also a solution to \eqref{eq:bb}. Conversely, if $\{(A^\tp)_i, \Lambda_i, x_i, y_i, \gamma_i\}$ solves \eqref{eq:bb} and satisfies $\Lambda\gamma=p$, then $(\Lambda, \gamma, A)$ solves \eqref{eq:z}.
\begin{equation}
    \begin{aligned}
        \max_{\gamma_i, x_i} & \sum_{i} \gamma_i x_i \\
        \text{subject to } & \sum_i \gamma_i=1, \gamma_i>0, \\
        & \{(A^\tp)_i, \Lambda_i, x_i, y_i\} \text{ solves } \eqref{eq:bilinear}
    \end{aligned}\label{eq:bb}
\end{equation}
\end{theorem}
\begin{proof}
    \eqref{eq:z} $\Rightarrow$ \eqref{eq:bb} is straightforward. Conversely, consider a solution $\{a_i, \lambda_i, x_i, y_i, \gamma_i\}$ to \eqref{eq:bb}. First, as $\gamma_i>0$, $a_i^\tp  V \lambda_i \geq a_i'^\tp V \lambda_i \Leftrightarrow a_i^\tp  V \lambda_i \gamma_i\geq a_i'^\tp V \lambda_i \gamma_i, i\in [N]$, implying $\operatorname{diag}(AVZ) \succeq  \operatorname{diag}(A'VZ)$, $Z=\Lambda\diag(\gamma)$. Due to the constraints $\sum_i \gamma_i=1$, $\gamma_i>0$, and $\Lambda \gamma =p$, we have $Z\in \mathcal{Z}$. Since $\gamma$ is optimal, $\sum_i a_i^\tp U \lambda_i\gamma_i\geq \sum_i a_i^\tp U \lambda_i\gamma_i'$, for any $\gamma'\in \Delta([M]), \gamma_i'>0$. Finally, the belief dominance gives $\sum_i a_i^\tp U \lambda_i\gamma_i'\geq \sum_i a_i^\tp U \lambda_i'\gamma_i',\forall \lambda'_i\in \Delta([N])$. Combining these inequalities, $\sum_i a_i^\tp U \lambda_i\gamma_i\geq \sum_i a_i^\tp U \lambda_i\gamma_i' \geq \sum_i a_i^\tp U \lambda_i'\gamma_i'$, i.e.,$ \operatorname{Tr}(AUZ)\geq \operatorname{Tr}(AUZ')$.  
\end{proof}

Before inspecting the tightness of the upper and lower bounds, we introduce a finite-time algorithm to find the exact solution to \eqref{eq:bb}. The first step is to identify the feasible set characterized by \eqref{eq:bilinear}. One can see from the proof and \eqref{eq:nash-constraint} that the purpose of bilinear programming \eqref{eq:bilinear} is to enumerate all solutions of the bimatrix game $(U,V)$, which is equivalent to find all solutions to the following programming. Prior works \cite{muk78exact,gallo77cutting, avis10irs-nash} have established finite-time algorithms (with exponential complexity) to enumerate the exact solutions. An online solver, named \texttt{lrs} (lexicographic reverse search), is offered by \cite{avis10irs-nash}.  
\begin{equation}
    \begin{aligned}
    \max_{a,\lambda, x,y} &  a^\tp U \lambda + a^\tp V \lambda -x-  y\\
    \text{subject to }&  U^\tp a \preceq x \mathds{1}, V \lambda \preceq y \mathds{1},\\
    & a\geq 0, \lambda\geq 0, a^\tp \mathds{1}=1, \lambda^\tp \mathds{1}=1.
\end{aligned} \label{eq:ne-bimat}   
\end{equation}
Consider the binary communication game in \Cref{exam:bd-pbe}. \texttt{lrs} returns three solution tuples represented as $(a, \lambda, x)$: $\{(1/2, 1/2), (1/3, 2/3), 1/3\}$, $\{(0, 1), (0,1), 1/2\}$, and $\{(1,0), (1,0), 1\}$. We now turn to the bilinear programming in \eqref{eq:bb}. Since $N=2$, we only need to keep two solution tuples of \eqref{eq:ne-bimat} as the feasible region. As $1/3<1/2<1$, it is natural to drop the solution $\{(1/2, 1/2), (1/3, 2/3), 1/3\}$ and keep the other two. In this case, the belief matrix becomes $\Lambda=I$, and hence, $\gamma=p=(1/2, 1/2)$. Since the feasible set for the variable $\gamma$ is a singleton, then the optimal solution is exactly the truth-telling strategy in \Cref{exam:bd-pbe}. Note that \texttt{lrs} only returns extreme equilibria \cite{avis10irs-nash}, which explains why \Cref{exam:bd-pbe} presents a continuum of PBE, whereas \texttt{lrs} only gives the above three. Yet, this technical nuance does not affect the objective value in \eqref{eq:bb}. In addition to the running example above, a more detailed discussion is presented in the arXiv version.

\textbf{The Tight Upper Bound} To evaluate the PoT, we first transfer the SPE in \eqref{eq:spe-matrix} into the posterior belief space, which is given by the following programming:
\begin{equation}
    \begin{aligned}
        \max_{\{\gamma_i, \lambda_i, a_i\}} & \sum_{i}a_i^\tp U\lambda_i \gamma_i\\
    \text{subject to } & a_i^\tp V \lambda_i \geq a_i'^\tp V \lambda_i,\\
    & a_i^\tp \mathds{1}=a_i'^\tp \mathds{1}=1, a_i,a_i'\geq 0,\\
    & \sum_{i}  \lambda_i \gamma_i=p.
    \end{aligned}\label{eq:spe-belief}
\end{equation}
We remark that \eqref{eq:spe-belief} is proposed mainly for analysis purposes. A linear programming approach is available for the computation purpose, using the revelation principle \cite{dugmi16aby,tao22bp}.   Comparing \eqref{eq:spe-belief} and \eqref{eq:nash-constraint}, one can see that PBE admits one more constraint regarding $\lambda$. Hence, $\text{PoT}\leq 1$. The following introduces a special class of communication games, referred to as strictly Bayesian-posterior competitive games, for which we prove that the upper bound is attained.  
\begin{definition}[Strictly Bayesian-posterior competitiveness]
\label{def:compete}
A game with payoffs $(U,V)$ is strictly Bayesian-posterior competitive if for all $a, a'\in \Delta(A)$ , $\lambda, \lambda'\in \Delta(\Omega)$, $a^\tp U\lambda -a'^\tp U \lambda' $ and $a'^\tp V \lambda'-a^\tp V \lambda$ have the same sign.
\end{definition}
\Cref{def:compete} extends the notion of strict competitiveness \cite{adler09compete} in normal-form games to communication games. The interpretation is that if one player benefits from changing from one outcome (belief-action pair) to another, the other must suffer. \cite{adler09compete} gives a simpler characterization of strict competitive games, which is helpful in our analysis.
\begin{lemma}[\cite{adler09compete}]
    A game  $(U,V)$ is strictly competitive if and only if there exist scalar $c,d,e,f$, with $c>0,e>0$ such that $cU+dJ_{K\times M} =-eV +f J_{K\times M}$.
\end{lemma}
\begin{theorem}[Tightness of the Upper Bound]
\label{thm:upp}
    Assuming that belief-dominant equilibrium exists for some strictly Bayesian-posterior competitive game, then $\text{PoT} =1$. 
\end{theorem}
\begin{proof}
    It suffices to prove the zero-sum case $U=-V$. In covert signaling, the constraints \eqref{eq:nash-constraint} in PBE become saddle-point constraints: for any equilibrium solution $a_i^*, \lambda_i^*$, $a_i^{*\tp} U \lambda_i'\leq a_i^{*\tp} U \lambda_i^{*\tp}\leq a_i'^{\tp} U\lambda_i^*$. Due to the interchangeability of mixed saddle-point strategies and the uniqueness of the saddle point value, the optimal value is $U^{CS}=a_i^{*\tp} U \lambda_i^*=a_j^{*\tp} U \lambda_j^*=a_i^{*\tp} U \lambda_j^*$ for any $i,j\in [N]$.  Hence, $U^{CS}=\sum_{i}a_i^{*\tp} U \lambda_i^{*} \gamma_i^*=\sum_{i} a_1^{*\tp} U \lambda_i^{*} \gamma_i^*=a_1^*U p=\max_i (Up)_i$.

In overt persuasion, the constraint in SPE \eqref{eq:spe-belief} turns into $a_i^\tp  U \lambda_i \leq a_i'^\tp U \lambda_i, \forall a_i'$, implying that for any fixed $l\in [K]$, $a_i^\tp  U \lambda_i \leq e_{l}^\tp U \lambda_i$, for $i\in [N]$ . Hence, for any $\gamma_i, 
    \lambda_i$, we have $\sum_{i}a_i^\tp  U \lambda_i \gamma_i \leq \sum_{i}e_l^\tp U \lambda_i\gamma_i=e_l^\tp U p\leq \max_i (Up)_i$, showing that $U^{OP}\leq U^{CS}$. Since by default $U^{OP}\geq U^{CS}$ (fewer constraints), then the two are equal.       
\end{proof}
\begin{remark}\label{remark:uninformative}
Even though the existence assumption is imposed in \Cref{thm:upp}, this is in fact dispensable, as one can construct a strictly competitive game that admits a belief-dominant equilibrium. The construction is inspired by the proof above, where we show that the equilibrium pair $a_i, \lambda_i$ constitutes a saddle point of $U$. For given utility matrices $U=-V$, denote by $(a,\lambda)$ an arbitrary saddle point of $U$. Then, let $A=\1 a^\tp $, $\Lambda=\lambda\1^\tp$, $\gamma=\frac{1}{N}\1$. Finally, define the prior as $p=\lambda$. By construction, $(\Lambda, A, \gamma)$ is a belief-dominant PBE to the communication game $\langle U, V, p\rangle$. The interpretation is that in this strictly competitive game, the receiver believes that every signal bears the sender's strategic intention that is in conflict with its own. Consequently, the receiver would ignore every signal from the sender: the posterior belief $\lambda$ equals the prior. Aware that its signal would be ignored anyway, the sender implements an uninformative information structure $\Pi=\frac{1}{N}\1\1^\tp$: a uniform distribution over the signal space independent of the state. 
\end{remark}

\section{Quadratic Games and Numerical Examples}\label{sec:quad}
 \eqref{eq:z} and \eqref{eq:bb} give a TSB characterization of PBE, which is instrumental in showing the tightness of the upper bound in finite games.  However, the bilinear programming does not reveal the tightness of the lower bound. This section presents a case study on the PoT in a particular continuous game: quadratic communication game (QCG). The PoT in QCG can be arbitrarily close to zero, implying the tightness of the lower bound.   

QCG consists of continuous state, signal, and action spaces: $\Omega=\mathcal{S}=\mathcal{A}=[0, 1]$, as well as quadratic utilities: $u(a, \omega)=-(a-\omega-b)^2, v(a,\omega)=-(a-\omega)^2$. The bias term $b>0$ denotes the misalignment between two players' interests: as $b\rightarrow 0$, the two are more aligned.  The receiver tries to guess where the actual state is (minimizing the error) based on the signal from the sender, who tries to mislead the receiver to somewhere else (specified by the offset $b$). The prior is the uniform distribution denoted by $p=\operatorname{unif}(0,1)$. 

\textbf{PBE in Signaling} An important finding in \cite{crawford82signaling} is that all PBE in QCG are partition equilibria. Given a constant $b>0$, there exists a positive integer $N(b)= \floor{-\frac{1}{2}+\frac{1}{2}(1+\frac{2}{b})^{1/2}}$ ($\floor{\cdot}$ is the ceiling function) such that there exists a PBE for every $N\in  [N(b)]$. The equilibrium information structure is in the form of partition signaling: for any $N\in [N(b)]$, there exists a sequence $0=k_0<k_1<\ldots<k_N=1$ such that 
\begin{equation}
    \pi(\cdot|\omega)=\operatorname{unif}(k_i,k_{i+1}),  \IF \omega\in (k_i, k_{i+1}).\label{eq:parti}
\end{equation}

In the partition equilibria, the sender randomly samples a signal from the sub-interval within which the true state falls, telling a half-truth to the receiver. One can clearly see from \eqref{eq:parti} that the more nearly players' interests coincide (the closer $b$ approaches zero), the finer partition there can be (the larger $N(b)$). On the contrary, as $b\rightarrow\infty$, $N(b)$ eventually falls to unity, and the sender would transmit uninformative signals to the receiver. Direct calculation shows that the watershed is $1/4$: as $b$ exceeds $1/4$, $\pi(\cdot|\omega)=p$ for all $\omega$. For the rest of this section, we assume $b\in (0, 1/4)$. 

We now turn to the sender's PBE payoff, i.e., $U^{CS}$. \cite[Theorem 1]{crawford82signaling} states that under the information structure \eqref{eq:parti}, $U^{CS}=-\sum_{i\in[N]} \Var_{\operatorname{unif}(k_{i-1},k_{i})}-b^2$, where $\Var_{\operatorname{unif}(k_i,k_{i+1})}$ denotes the variance of the uniform distribution over $[k_{i-1}, k_i]$. The equilibrium partition of number $N\in [N(b)]$, as shown in \cite{crawford82signaling}, is $k_i=i/N+2bi(i-N)$, $i\in [N]$. Hence, a direct calculation gives 
\begin{equation}
    U^{CS}=-\frac{1}{12N^2}-\frac{b^2(N^2-1)}{3}-b^2.\label{eq:ucs}
\end{equation}

\textbf{SPE in Persuasion}
The calculation of SPE in QCG rests on the backward induction in \eqref{eq:spe}. Given a posterior belief $\lambda$, the best response is $\hat{a}(\lambda)=\argmax_a \E_{\omega\sim \lambda}[-(a-\omega)^2]=\E_\lambda [\omega]$. The sender's expected payoff under $\lambda$ is $\hat{u}(\lambda)=\E_{\omega\sim\lambda} u(\hat{a}(\lambda), \omega)= -\Var_{\lambda}-b^2$, where $\Var_\lambda$ denotes the variance of the random variable $\omega\sim \lambda$. Finally, $U^{OP}$ is the optimal value of the following problem:
\begin{equation}
	\begin{aligned}
	\max_{\tau\in \Delta(\Delta([0,1]))} \quad & \E_\tau [-\Var_{\lambda}-b^2]\\
	\st & \int_{\lambda\in \operatorname{supp}(\tau)}\lambda \tau(d\lambda )=p.
\end{aligned} \label{eq:spe-uniform}
\end{equation}
 If we choose $\lambda$ as a Dirac measure $\delta(\omega)$, for $\omega\in [0,1]$, then the variance term is zero. Hence, the optimal value is $U^{OP}=-b^2$. The interpretation is that the sender opts for a truth-telling strategy, i.e., $s=\omega$, even though the incentive bias exists $b>0$. Consequently, the receiver's belief collapses to its true state $\lambda(\cdot|\omega)=\delta(\omega)$. Mathematically, this truth-telling equilibrium is due to the fact that $\hat{u}(\lambda)$ is convex in $\lambda$, and the reader is referred to \cite[Section V]{kamenica11BP} for more details, where authors consider a lobby game similar to our setting.      

\textbf{The Tight Lower Bound} With all the results above, we now address the lower bound of PoT. Note that for the simplicity of exposition, we construct non-positive utilities in QCG, violating the non-negativity assumption in \Cref{sec:model}. If  blindly computing $\frac{U^{CS}}{U^{OP}}$, one would arrive at the opposite conclusion. Therefore, we prove that PoT converges to zero by showing that $U^{OP}$ converges to zero (the maximum) at a higher order than $U^{CS}$ does, as $b\rightarrow 0$.  This higher-order convergence indicates that OP significantly outperforms CS.  
\begin{theorem}[Tightness of the Lower Bound]
\label{thm:lower}
Consider  the quadratic communication  game of the incentive bias $b>0$, PoT converges to $0$, as $b$ tends to $0$. 
\end{theorem}
\begin{proof}
    It suffices to prove that the convergence order of $U^{CS}$ is strictly less than 2. First, note that for PBE payoff in \eqref{eq:ucs} is non-unique, as every $N\in [N(b)]$ leads to a PBE. Therefore, we consider the PBE with the maximum payoff achieved when $N=N(b)$, and the resulting PBE is Pareto-superior to all other equilibria \cite{crawford82signaling}. Then, notice that $N=N(b)\sim O(b^{-1/2})$, we arrive at $N^{-2}\sim O(b)$ and $N^2b^2\sim O(b)$; i.e., $U^{CS}$ in \eqref{eq:ucs} is of $O(b)$. 
\end{proof}
\begin{remark}[Half-Truth still Hurts.]
One interpretation of \Cref{thm:lower} is that the opaque information disclosure, compared to the transparent, creates ``friction'' in information transmission. As one can see from the partition equilibria in \eqref{eq:parti}, the informativeness of the signaling is reflected by the width of each sub-interval $d_i=k_{i+1}-k_i$. The finer the partition is, the smaller $d_i$ is, and the more confident the receiver is about the true state. As $b\rightarrow 0$, and $d_i$ shrinks, the half-truth gets closer to the truth. Yet, the half-truth still hurts: the signal bears randomness (unlike the deterministic signal in OP), even though the two players' interests coincide. The resulting $U^{CS}$ exhibits a first-order convergence.               
\end{remark}

\section{Conclusion}
This work has introduced the notion of \emph{price of transparency} (PoT) to quantify the cost or benefit of information disclosure in strategic interactions. It allows for the assessment of the sender's tradeoffs when adhering to ethical standards that require transparency in information disclosure. We have observed that counterintuitively, choosing transparency can yield a payoff no less than that under an opaque information structure, with PoT values ranging between $0$ and $1$.  We have developed a two-stage-bilinear programming approach \eqref{eq:z} using Bayesian plausibility to solve for the perfect Bayesian equilibrium. Furthermore, this programming approach has enabled us to show the upper bound is attainable for strictly Bayesian-posterior competitive games. Additionally, we have constructed quadratic games where PoT can be arbitrarily close to $0$. The tight lower bound implies that the sender can be plagued by the lack of transparency. 
\bibliographystyle{ieeetr}
\bibliography{potref}
\end{document}